\documentclass[a4paper, 11pt]{amsart}
\UseRawInputEncoding
\usepackage{amsmath, amssymb, amsthm}
\usepackage[english]{babel}
\usepackage[T1]{fontenc}
\usepackage{graphicx}
\usepackage{rotating}
\usepackage{array,multirow}
\newcommand{\R}{\mathbb R}
\newcommand{\N}{\mathbb N}

\def\be#1\ee{\begin{equation}#1\end{equation}}
\newcommand{\fer}[1]{(\ref{#1})}
\newtheorem{theorem}{Theorem}[section]
\newtheorem{lemma}[theorem]{Lemma}

\newtheorem{remark}[theorem]{Remark}

\numberwithin{equation}{section}

\newcommand{\bq}{\begin{equation}}
\newcommand{\eq}{\end{equation}}
\newenvironment{equations}{\equation\aligned}{\endaligned\endequation}
\def\bqa{\begin{eqnarray}}
\def\eqa{\end{eqnarray}}

\def\e{\epsilon}

\def\t{\tau}

\newcommand{\bd}{\begin{displaymath}}
\newcommand{\ed}{\end{displaymath}}
\newcommand{\ba}{\begin{eqnarray}}
\newcommand{\ea}{\end{eqnarray}}

\def\ff{\widehat f}

\def\gg{\widehat g}
\def\hh{\widehat h}

\def\N{\mathbb{N}}
\def\R{\mathbb{R}}

\newcommand{\setR}{\mathbb{R}}

\theoremstyle{plain}

\title{On  Fourier-based inequality indices}

\author{Giuseppe Toscani}

\thanks{Department of Mathematics  of the University of Pavia, and IMATI CNR, Pavia, Italy.  \\ e.mail: {\tt giuseppe.toscani@unipv.it}. 
}

\date{\today}

\begin{document}
\maketitle

\begin{center}\small
\parbox{0.85\textwidth}{

\textbf{Abstract.} 
Originally developed for measuring the heterogeneity of wealth measures, inequality indices are quantitative scores that take values in the unit interval, with the zero score characterizing perfect equality.  In this paper, we draw attention to a new inequality index, based on the Fourier transform, which exhibits a number of interesting properties that make it very promising in applications. As a by-product, it is shown that other inequality measures, including Gini and Pietra indices can be fruitfully expressed in terms of the Fourier transform, which allows to enlighten properties in a new and simple way.


\medskip

}
\end{center}

\section{Introduction}

Denote by $P_s(\R)$, $s \ge 1$, the class of all probability measures $F$ on the Borel subsets of $\R$  such that
\[
m_s(F)  = \int_{\R} |x|^s dF(x) < + \infty.
\]
Further, denote by $ \tilde P_s (\R)$ the class of probability measures $F \in P_s(\R)$ which possess a positive mean value
\[
m(F) =  \int_{\R} x \, dF(x) >0, 
\]
and with $P_s^+(\R)$ the subset of probability measures $F \in P_s(\R)$ such that $F(x) = 0$ for $x \le 0$. 
Let $\mathcal F_s$ be the set of Fourier transforms  
\[
\ff(\xi) = \int_{\R}  e^{-i\xi x} \, dF(x).
\]
of probability measures $F$ in $\tilde P_s(\R)$.
On $\mathcal F_s$ we introduce an inequality index, named $T(F)$, given by the formula
\be\label{ine-T} 
T(F) = \frac 12 \sup_{\xi \in \R} \left| \ff(\xi) - \frac{\ff'(\xi)}{\ff'(0)} \right|,
\ee
In definition \fer{ine-T},  $\ff'(\xi)= \ff'_\xi(\xi)$ denotes the derivative of the Fourier transform $\ff(\xi)$ with respect to its argument $\xi$. Indeed,  $F \in P_s(\R)$ implies that  $\ff(\xi)$ is  continuously differentiable on the entire real line.  

It is immediate to show that the functional $T(F)$ is invariant with respect to the scaling (dilation)
\[
F(x) \to F(cx), \quad c >0.
\]

Let $F \in P_s^+(\R)$. Since $|\ff(\xi)| \le \ff(0) = 1$, and 
\[
\ff'(\xi) = -i \int_{\R_+}  x e^{-i\xi x} \, dF(x),
\]
so that $\ff'(0)= -i\, m(F) =-i\, m$, and $|\ff'(\xi)| \le |\ff'(0)| = m$, it is immediate to conclude, by the triangular inequality, that  $T(F)$ satisfies the bounds
\be\label{bounds}
0 \le T(F) \le 1,
\ee
and $T(F) = 0$ if and only if $\ff(\xi)$ satisfies the differential equation
\[
 \ff'(\xi) = {\ff'(0)}\, \ff(\xi),
\]
with $\ff(0) = 1$, so that the unique solution is given by $\ff(\xi) = e^{-im\xi}$, namely by the Fourier transform of a Dirac delta function located in the mean value $x =m(F)>0$. Note however that, even if the functional $F$ is defined in the whole class $\tilde P_s(\R)$,  the upper bound is lost if the probability measure $F \notin  P_s^+(\R)$, since in this case the inequality $|\ff'(\xi)/\ff'(0)| \le 1$ does not hold. 

 In the following, we will show that the functional $T(F)$ is a measure of inequality which satisfies most of the properties required to be  a good measure of sparsity and/or heterogeneity \cite{HR}.  
 The interest in having a measure of inequality  based on the Fourier transform, such as $T(F)$, is twofold. On the one hand, it is very simple to calculate the value taken by this measure at probability distributions for which the characteristic function is explicitly available. This is the case, among others, of the Poisson distribution and, for probability measures defined on the whole real line $\R$, of the stable laws.  On the other hand, in the case of dealing with a discrete probability measure, the use of the Fourier transform makes it possible to develop very fast computational procedures \cite{Au1,Au2}.

 As recently discussed in \cite{Ban,Eli,Eli3}, the challenge of measuring the statistical heterogeneity of  measures arises in most fields of science and engineering. In economics and social sciences size measures of interest are wealth measures, and in the context of wealth measures many inequality indices have been introduced \cite{BL,Cou,Cow,HN} . Specifically, inequality indices quantify the socio-economic divergence of a given wealth measures from the state of perfect equality. In economics, the most used measure of inequality is the Gini index \cite{Gini1,Gini2}. However, although it has had an economic origin, the use of the Gini index has not been limited to wealth alone \cite{HR}. 
 
 A second important index of inequality, still introduced in economics, is the Pietra index \cite{Pie}. As discussed in \cite{Eli2}, the Pietra index is an elemental measure of statistical heterogeneity which has a number of properties that render it not only an alternative to the popular Gini index, but rather, a far more natural and meaningful quantitative tool for the measurement of egalitarianism, and, consequently, for the measurement of statistical heterogeneity at large.

In addition, other indices have been introduced so far.   An alternative to Gini index was introduced by Bonferroni in 1930 in a textbook for students at Bocconi University in Milan \cite{Bon}. The main properties and representations  of Bonferroni index and its connections with the index of Gini and other measures were studied in \cite{Eli3}.  Also, it is important to mention the Kolkata index, first introduced in \cite{Gho} as a measure of  inequality,  whose connections with the Gini and Pietra indices have been studied in \cite{Ban, Ban2}. 

An indispensable tool for measuring income and wealth  inequality is the Lorenz function and its graphical representation, the Lorenz curve \cite{Lor}. The Lorenz curve plots the percentage of total income earned by the various sectors of the  population, ordered by the increasing size of their incomes.
The Lorenz curve is typically represented as a curve in the unit square of opposite vertices in the origin of the axes and the point $(1,1)$,  starting from the origin and ending at the point  $(1,1)$. 
  
The diagonal of the square exiting the origin is the line of perfect equality, representing a situation in which all individuals have the same income. Since the diagonal is the line of perfect equality, we can say that
 that the closer the Lorenz curve is to the diagonal, the more equal is the
 distribution of income. 
 
 This idea of \emph{closeness} between the line of perfect equality and the Lorenz curve can be expressed in many ways, each of which gives rise to a possible measure of  inequality.
  Thus, starting from the Lorenz curve, several indices of  inequality were defined, including the Gini index. Various indices were obtained by looking at the maximal distance between the line of perfect equality and the Lorenz curve, either horizontally or vertically, or alternatively parallel to the other diagonal of the unit square \cite{Eli}.
 
 From a certain point of view, the measure of inequality defined by \fer{ine-T} has many points of contact with the  inequality measures obtained from the Lorenz curve through the concept of \emph{maximum distance}.
 
In fact, the index \fer{ine-T} expresses the maximum value of the modulus of the difference between the Fourier transforms of a probability measure $F$ of positive mean $m$ and its derivative normalized by the mean. 
In the economic context, the closer the Fourier transform of the probability measure is to its derivative normalized by dividing it by the mean, the more equal is the distribution of income. In other words, the line of perfect equality in the Lorenz square is here substituted by the Fourier transform of a Dirac delta function located in a point different from zero.

It is interesting to note that, as will become clear from the examples, the maximum value is usually taken in the finite interval $(-2\pi, 2\pi)$.

Despite the enormous amount of research illustrating the fields of application of inequality indices, the use of arguments based on Fourier transforms appears rather limited. In particular, although the Gini index can be easily expressed in terms of the Fourier transform, at least to our knowledge, its expression in Fourier has never been considered in applications.
The same conclusion can be drawn for the Pietra index, whose expression in Fourier transform is very useful to understand its nature, and to introduce from that other Fourier-based measures of inequality, including the one considered in this paper.  

It is worth mentioning that, unlike the classical Gini and Pietra indices, neither the Bonferroni index nor the Kolkata index seem to be expressible in closed form in terms of the Fourier transform.

Before studying the new index $T(\cdot)$ defined in \fer{ine-T} and listing its properties, we will begin with a brief introduction to the use of the Fourier transform to express the classical Gini and Pietra indices.  This will be done in  Section  \ref{sec:Gini-Pietra}. As we shall see, the use of Fourier transform allows to clarify the functional setting where these indices live.

Next,
Section \ref{sec:new}  will  be devoted to the study of the main properties of the new inequality measure.  Various examples will be collected in Sections \ref{sec:examples} and  \ref{sec:kinetic}.

\section{A Fourier approach to Gini and Pietra indices} \label{sec:Gini-Pietra}

\subsection{A Fourier-based expression of Gini index} 
 In the rest of the paper, for any fixed constant  $a >0$, we will denote by $F_a(x)$ the Heaviside step  function defined by
 \be\label{he}
 F_a(x) : = \left\{
  \begin{array}{cc}
0 &x< a\\
  1 & x \ge a
  \end{array}
 \right.
 \ee
Clearly, $F_a(x)$ is the cumulative measure function of a random variable which is almost surely equal to $a$. It belongs to 
$P_s(\R)$ for any $s \ge 1$, and $m(F_a) = a$. 

To obtain an explicit expression in Fourier transform for the Gini index, which admits many equivalent formulations \cite{Xu}, we will resort to its well-known form in terms of a continuous probability measure. For a probability measure $F \in P_s^+(\R)$ with mean $m$, the Gini index is defined by the formula
\be\label{Gin}
G(F) = 1 - \frac 1m\int_{\R_+} (1- F(x))^2\, dx
\ee
Since $F \in P_s^+(\R)$, $F(x) = 0$ for $x \le 0$. Hence, resorting to the definition of the Heaviside step function  $F_0(x)$, we have the identity
\[
\int_{\R_+} (1 - F(x))^2 \, dx =\int_\R |F_0(x) - F(x)|^2 \, dx.
\]

For any given pair of probability measures $F,G \in P_s(\R)$,  Parseval formula implies
\be\label{Par}
\int_\R |F(x) - G(x)|^2 \, dx = \frac 1{2\pi} \int_\R |\widehat F(\xi) - \widehat G(\xi)|^2 \, d\xi,
\ee
 where $\widehat F$ and $ \widehat G$ are the Fourier transforms of the probability measures $F,G$. It holds
 \be\label{trans}
 \widehat F(\xi) - \widehat G(\xi) = \frac{\ff(\xi) -\gg(\xi)}{i\xi}.
 \ee
 Indeed, considering that $F(-\infty) - G(-\infty) = F(+\infty) - G(+\infty) =0$, integration by parts gives
 \begin{equations}\nonumber
& \int_\R (F(x) -G(x)) e^{-i\xi x} \, dx = \int_\R (F(x) -G(x))\frac d{d\xi}\left( \frac{e^{-i\xi x}}{-i\xi}\right)  \, dx = \\
& \left[ (F(x) -G(x)) \frac{e^{-i\xi x}}{-i\xi} \right]_{-\infty}^{+\infty} + \int_\R  \frac{e^{-i\xi x}}{i\xi} \, d(F(x) -G(x)) = \frac{\ff(\xi) -\gg(\xi)}{i\xi}.
 \end{equations}
 Consequently, we have the identity
 \be\label{Par2}
\int_\R |F(x) - G(x)|^2 \, dx = \frac 1{2\pi} \int_\R \frac{|\ff(\xi) -\gg(\xi)|^2}{\xi^2} \, d\xi.
\ee

Therefore, for any probability measure $F\in P_s^+(\R)$,  Gini index has  a simple expression in Fourier transform, given by
\be\label{Gin-F}
G(F) = 1 - \frac 1{2\pi m} \int_\R \frac{|1 -\ff(\xi)|^2}{\xi^2} \, d\xi.
\ee

\begin{remark}\label{rem:1}  For a given constant $q>0$, let  $\dot H_{-q}$  denote the homogeneous Sobolev space of fractional order with negative index ${-q}$, endowed with the norm
 \be\label{h2}
  \| h\|_{\dot H_{-q}} = \int_\R |\xi|^{-2q}|\widehat h(\xi)|^2 \, d\xi.
 \ee
Then, the variable part of the Gini index
coincides with the scaling invariant distance  between the probability measure $F$ and the the heaviside step function $F_0$ in the homogeneous Sobolev space  $\dot H_{-1}$.
\end{remark}

\begin{remark}\label{rem:2}  Considering that the value zero in \fer{Gin-F} is obtained when $\ff(\xi) = e^{-im\xi}$, namely when $F = F_m$, we can rewrite Gini index as
\be\label{G-diff}
G(F) =  \frac 1{2\pi m} \left[ \int_\R \frac{|1 - e^{-im\xi} |^2}{\xi^2} \, d\xi - \int_\R \frac{|1 -\ff(\xi)|^2}{\xi^2} \, d\xi\right].
\ee
\end{remark}

\subsection{Another Fourier-based inequality measure}\label{sec:H-dist}

Expression \fer{G-diff} suggests considering a related expression in which the dispersion of the probability measure $F$ of mean value $m>0$ coincides with its scale invariant $\dot H_{-1}$--distance from the Heaviside step function $F_m$ with the same mean value $m$. We define
\be\label{H-F}
H(F) =  \frac 1{2\pi m} \int_\R \frac{|\ff(\xi) - e^{-im\xi} |^2}{\xi^2} \, d\xi.
\ee
At difference with Gini index, which requires $F \in P_s^+(\R)$, the inequality measure $H(F)$ is well-defined for any measure $F \in \tilde P_s(\R)$.

It is interesting to remark that, similarly to Gini index, the inequality measure $H(F)$, for $F \in P_s^+(\R)$, is bounded above by $1$.  This property is shown in the following

\begin{lemma}\label{bound-H} Let $F \in P_s^+(\R)$ be a probability measure  of mean value $m$. Then
\be\label{b-H}
0 \le H(F) < 1.
\ee
\end{lemma}

\begin{proof}
Let $F \in P_s^+(\R)$ be a probability measure  of mean value $m$. Thanks to Parseval formula, the value of th expression \fer{H-F} in the Fourier space coincides with the value
\be\label{Par-3}
H(F) =  \frac 1{ m} \int_{\R_+} |F(x) -F_m(x) |^2\, dx.
\ee
The simplest case in which we can explicitly evaluate $H(F)$ is when $F(x)\in P_s^+(\R)$  is the measure function of a random variable $X$ taking only two non-negative values, that is, for $0<p<1$
\[
P(X= m-a) = p; \quad P(X= m+b) = 1- p, \quad a,b >0; \quad a\le m.
\]
Since $X$ has mean value $m$, $a$ and $b$ are related to $p$ by the relation
\be\label{aa}
pa = (1-p)b.
\ee
In this case, it is a simple exercise to verify that
\[
\frac 1{ m} \int_{\R_+} |F(x) -F_m(x) |^2\, dx = \frac 1m\left[p^2 a + (1-p)^2 b\right],
\]
so that, thanks to \fer{aa}
\[
H(F) = \frac{pa}m.
\]
Therefore, since $a \le m$ and $p < 1$, we conclude with $H(F) < 1$. 

An interesting application of the previous expression is obtained  by assuming $a =m$, and $p =1-\e$, with $0<\e\ll 1$. 
In this case the random variable $X$ of mean value $m$ is such that
\[
P(X= 0) = 1-\e; \quad P\left(X= \frac m\e\right) = \e.
\]
In economics, this situation describes a population in which most of agents have zero wealth, while one small part possesses an extremely  high wealth, while maintaining the mean wealth fixed. In this case $H(F) = 1-\e$.

Let us now consider a random variable $X$ of mean value $m$ that takes three non negative values $x_1 <x_2< m< x_3$,
where
\[
P(X=x_k) = p_k, \quad k = 1,2,3; \quad p_1+p_2+p_3 =1.
\]
In this case
\[
\int_{\R_+} |F(x) -F_m(x) |^2\, dx = p_1^2(x_2-x_1) + (p_1+p_2)^2(m-x_2) + p_3^2 (x_3 -m).
\]
Let 
\be\label{medio}
x = \frac{p_1}{p_1 +p_2}\, x_1 + \frac{p_2}{p_1 +p_2}\, x_2.
\ee
Then, 
 $x \in (x_1, x_2)$, and 
\[
(p_1+p_2) x = p_1x_1+p_2x_2.
\]
Let $Y$ be the two valued random variable defined by
\[
P(Y = x) = p_1 +p_2; \quad P(Y = x_3) = p_3.
\]
Then, thanks to \fer{medio}, $Y$ has mean value $m$. Moreover, if we denote by $G$ the measure function of $Y$,
\[
\int_{\R_+} |G(x) -F_m(x) |^2\, dx = (p_1+p_2)^2(m-x) + p_3^2 (x_3 -m).
\]
Owing to \fer{medio} we obtain
\begin{equations}\nonumber
& (p_1+p_2)^2 (x_2-x) = (p_1+p_2)^2\left(x_2 -  \frac{p_1}{p_1 +p_2}x_1 + \frac{p_2}{p_1 +p_2}x_2\right) = \\
& p_1 (p_1+p_2) (x_2-x_1) \ge p_1^2(x_2- x_1).
\end{equations}
Consequently
\begin{equations}\nonumber
& p_1^2 (x_2-x_1) + (p_1+p_2)^2\left(m- x_2 \right) \le  (p_1+p_2)^2 (x_2-x) +  (p_1+p_2)^2\left(m- x_2 \right) =\\
& (p_1+p_2)^2 (m-x),
\end{equations}
that implies
\be\label{oo}
H(F) \le H(G) < 1.
\ee
The same conclusion holds if we consider a random variable $X$ of mean value $m$ that takes the three non negative values $x_1 <m < x_2 < x_3$, and we choose the value $x \in (x_2, x_3)$ like in \fer{medio}. The previous computations show that, by suitably choosing the point, we can built, starting from a random variable with three values, a random variable with two values, with the same mean and with a bigger value of the functional $H$, which by the previous computations is less than $1$. At this point, we can iterate the procedure and conclude that the upper bound in \fer{b-H} holds for the measure function $F \in P_s^+(\R)$ of any discrete random variable $X$, and finally for any $F \in P_s^+(\R)$. 
\end{proof}

The interest in having an inequality index that quantifies the statistical heterogeneity of probability measures defined on the whole real line $\R$ in terms of the Fourier transform is evident. As an example, let us compute the value of the functional $H$ 
for a Gaussian probability measure $F$ of mean $m>0$ and variance $\sigma^2$. Since the Fourier transform of the Gaussian density is given by
\be\label{fou-G}
\ff(\xi) = \exp\left\{ -i\, m\, \xi  -\frac{\sigma^2}2\xi^2\right\},
\ee
we easily obtain
\[
H(F) =  \frac 1{2\pi m} \int_\R \frac{\left( 1 -\exp\left\{ -\frac{\sigma^2}2\xi^2\right\} \right)^2}{\xi^2} \, d\xi.
\]
Integration by parts yields
\begin{equations}\nonumber
& \int_0^\infty \frac{\left( 1 -\exp\left\{ -\frac{\sigma^2}2\xi^2\right\} \right)^2}{\xi^2} \, d\xi = \left[-{\left( 1 -\exp\left\{ -\frac{\sigma^2}2\xi^2\right\} \right)^2}\frac 1{\xi}\right]_0^\infty +\\
& + \int_0^\infty \frac 1\xi  \left( 1 -\exp\left\{ -\frac{\sigma^2}2 \xi^2\right\} \right) \sigma^2\xi\exp\left\{ -\frac{\sigma^2}2 \xi^2\right\}  \, d\xi = \\
& \sigma \int_0^\infty  \left( e^{-x^2/2 } - e^{-x^2} \right)\, dx = \sigma \frac{ (\sqrt 2 -1)\sqrt\pi}2.
\end{equations}
Thus, for a Gaussian probability measure $F$ of mean $m>0$ and variance $\sigma^2$ we have the value
\be\label{ggg}
H(F) = \frac{ (\sqrt 2 -1)}{2\sqrt{2\pi}} \frac{\sigma}m,
\ee
namely a value proportional to the coefficient of variation $\sigma/m$, with an explicit constant strictly less than one.

\subsection{A Fourier-based expression of Pietra index} 
For a probability measure $F \in P_s^+(\R)$ with mean $m$, the Pietra index $P(F)$ \cite{Eli2,Pie} is defined by the formula
\be\label{Pie}
P(F) = \int_m^{+\infty}(1- F(x))\, dx.
\ee
As remarked in \cite{Eli2}, the definition \fer{Pie} seems to disregard the part of the measure below the mean. This, however, is not true, and \fer{Pie} makes use of the full information encapsulated in the probability law of the random variable $X$ of measure $F$.

There is a simple way to verify the previous assertion. Indeed, since
\[
m = \int_{\R_+}(1-F(x))\, dx,
\]
it holds
\begin{equations}\nonumber
&H(F) = \frac 1m\int_{\R_+} [F(x) -F_m(x) ]^2\, dx = \frac 1m\int_{\R_+} [1- F(x) -(1-F_m(x)) ]^2\, dx =\\
&  \frac 1m\int_{\R_+} [1- F(x)]^2\, dx + \frac 1m\int_{\R_+} [1- F_m(x)]^2\, dx -\frac 2m\int_{\R_+} [1- F(x)][1-F_m(x)]\, dx = \\ 
&\frac 1m\int_{\R_+} [1- F(x)]^2\, dx + 1 - \frac 2m\int_0^m[1- F(x)]\, dx =\\
&  \frac 1m\int_{\R_+} [1- F(x)]^2\, dx - 1 + \frac 2m\int_m^{+\infty}[1- F(x)]\, dx  = -G(F) + 2P(F).
\end{equations}
Hence, we have the identity
\be\label{Pie-new}
P(F) = \frac 12\left[G(F) + H(F) \right].
\ee
In other words, the Pietra index of a probability measure $F \in \tilde P_s^+$ is represented by the mean value of the two indices $G(F)$ and $H(F)$, where $H$ is defined in \ref{H-F},  identically weighted.

Resorting to the Fourier expressions of Gini and $H$ indices we then obtain for Pietra index the expression
\be\label{Pie-F}
P(F) = \frac 12\left[1 - \frac 1{2\pi m} \int_\R \frac{|1 -\ff(\xi)|^2}{\xi^2} \, d\xi + \frac 1{2\pi m} \int_\R \frac{|\ff(\xi)- e^{-im\xi} |^2}{\xi^2} \, d\xi  \right].
\ee

\begin{remark}  
The Fourier expression \fer{Pie-F} clarifies that  the Pietra index is obtained by taking into account at the same time the distances in $\dot H_{-1}$ of a probability measure in $P_s^+(\R)$ from the Dirac delta functions located in  zero, and, respectively in mean value $m$. From this point of view, the Pietra index appears as a well-balanced inequality index. This feature is hidden in the classical definition.
\end{remark}

\begin{remark} 
Since
\[
H(F) = 2P(F) - G(F),
\] 
the values of the inequality index $H(F)$ for  a large number of probability measures can be easily computed resorting to the tables  of values assumed by Gini and Pietra indices.  
\end{remark}

\begin{remark}
If one considers only one-dimensional discrete measures, the inequality index $H(F)$ defined by \fer{H-F} coincides with a particular case of the discrepancy function recently introduced in \cite{Au2}, where the discrepancy measures the distance in $L^2$ distance between the characteristic functions of two given discrete measures weighted by the function $k^2$, with $k = 1,2,\dots, N$. In this case, one of the two discrete measure is a Dirac delta function located in the mean value.
\end{remark}

\subsection{Towards new inequality indices}

A part from the scaling constant, the  functional $H(F)$, coincides with the square of the $L^2(\R)$--norm of the function 
\[
h(\xi) =  \frac{|\ff(\xi)- e^{-im\xi} |}{|\xi|}.
\]
It is immediate to verify that a further scaled invariant functional can be obtained by considering the $L^\infty(\R)$--norm of $h(\xi)$. This functional is given by
\be\label{sup-F}
H_\infty(F) = \frac 1{2m} \sup_{\xi\in \R}\frac{|\ff(\xi)- e^{-im\xi} |}{|\xi|}.
\ee
Resorting to the triangular inequality, we can easily conclude that, if $F \in \tilde P_s^+$, $H_\infty$ satisfies the standard bounds
\be\label{bbb}
0 \le H_\infty(F)< 1.
\ee
Indeed,  for $F \in P_s^+$ with mean value $m$
\[
H_\infty(F) = \frac 1m \sup_{\xi\in \R}\frac{|\ff(\xi)- e^{-im\xi} |}{|\xi|} \le  \frac 1m \sup_{\xi\in \R}\frac{|1- e^{-im\xi} |}{|\xi|} + \frac 1m \sup_{\xi\in \R}\frac{|1 -\ff(\xi)|}{|\xi|}.
\]
Now
\[
\frac 1m \sup_{\xi\in \R}\frac{|1- e^{-im\xi} |}{|\xi|} = \frac 1m \sup_{\xi\in \R}\frac{\sqrt{2(1-\cos \xi m)}}{|\xi|} = \lim_{\xi \to 0} \frac{\sqrt{2(1-\cos \xi m)}}{m|\xi|} = 1.
\]
Moreover, since by \fer{trans} 
\[
\frac{1- \ff(\xi)}{i\xi} = \int_\R (F_0(x) -F(x)) e^{-i\xi x} \, dx,
\] 
we obtain
\begin{equations}\nonumber
&\frac 1m \frac{|1 -\ff(\xi)|}{|\xi|} = \frac 1m \left| \int_\R (F_0(x) -F(x)) e^{-i\xi x} \, dx \right| \le \\ &\frac 1m \int_\R \left||(F_0(x) -F(x)) e^{-i\xi x}  \right| \, dx =
 \frac 1m  \int_{\R_+}(1 -F(x)) \, dx  = 1.
 \end{equations}

The functional $H_\infty$ is a particular case of  a metric for probability measures which have been used to study convergence to equilibrium for the Boltzmann equation. This is an argument that in kinetic theory of rarefied gases goes back to  \cite{GTW}, where convergence to equilibrium for the Boltzmann equation for Maxwell pseudo-molecules was studied in terms of a metric for Fourier transforms (cf. also \cite{Carrillo:2007, MaTo07, TV} for further applications). 

The metric  introduced in \cite{GTW} in connection with the Boltzmann equation for Maxwell molecules, was subsequently applied in various contexts, which include kinetic models for wealth measures \cite{PT13}, thus establishing a number of common points between kinetic modeling and inequality measures.  

For a given pair of random variables $X$ and $Y$ distributed according to $F$ and $G$ these metrics read
\be\label{me-inf}
 d_r(X,Y) = d_r(F,G) = \sup_{\xi\in\R} \frac{|\ff(\xi) -\gg(\xi)|}{|\xi|^r}, \quad r >0.
 \ee
As shown in \cite{GTW}, the metric $d_r(F,G)$ is finite any time the probability measures $F$ and $G$ have equal moments up to $[r]$, namely the entire part of $r\in \R_+$, or equal moments up to $r-1$ if $r \in \N$, and it is equivalent to the weak$\null^*$ convergence of measures for all $r >0$. Among other properties, it is easy to see \cite{GTW, PT13} that, for two pairs of  random variables  $X,Y$, where $X$ is independent from $Y$, and  $Z, \tilde Z$ ($Z$ independent from $\tilde Z$),  and any  constant $c$ 
 \begin{equations}\label{sca1}
 & d_r(X+Y, Z+\tilde Z)\le d_r(X,Z) + d_r(Y,\tilde Z) \\
 & d_r(c X, c Y) =  |c|^r d_r(X,Y).
 \end{equations}
These properties classify $d_s$ as an ideal probability metric in the sense of Zolotarev \cite{Zolo}. 
Properties of $H_\infty(F)$ can be easily extracted from \fer{sca1} considering that, if $X$ is a random variable with probability measure $F$ of mean value $m$
\[
H_\infty (F) = H_\infty(X) = \frac 1{2m} d_1(F, F_m)
\]
In particular, the second property in \fer{sca1} implies the scaling invariance of $H_\infty$.

Also the first inequality in \fer{sca1} implies that, for any pair of independent variables $X$ and $Y$, with means $m_X$ (respectively $m_Y$ ), by choosing $Z$ and $\tilde Z$ with probability measures $F_{m_X}$ (respectively $F_{m_Y}$)
\be\label{super}
H_\infty(X+Y) \le \frac{m_X}{m_X+m_Y} H_\infty(X) + \frac{m_Y}{m_X+m_Y} H_\infty(Y),
\ee
namely a property of sub-additivity for convolutions. Moreover, if $Y$ is distributed with probability measure $F_{m_Y}$, Inequality \fer{super} gives
\be\label{case}
H_\infty (X+Y) \le \frac{m_X}{m_X+m_Y} H_\infty(X) < H_\infty (X).
\ee
Inequality \fer{case} is a typical feature of sparsity measures, which translates to the case of a continuous variable the property that adding a constant to each coefficient decreases sparsity \cite{HR}.

In view of its properties, the functional $H_\infty(\cdot)$ appears to be a good measure of inequality. Unlikely, the computation of the values of $H$ for most probability measures is cumbersome. In particular, it seems not possible to explicitly compute the value of $H_\infty(X)$ in the simplest case in which the variable $X$ takes only two positive values. Consequently, we can not evaluate if, for a given $\e \ll1$, there exists a probability measure with index $1-\e$. Indeed, as we saw in Section \ref{sec:H-dist}, this basic property follows from the analysis of the values taken by the inequality index at two-valued random variables.

  The  upper bound $1$ can be found also by resorting to Lagrange theorem. Indeed,  since $F \in P_s(\R)$  the function  $h(\xi) = |\ff(\xi)- e^{-i m \xi }|$ is  continuously differentiable on the entire real line, and satisfies $h(\xi=0) = 0$. 

  Therefore, by Lagrange theorem, for any given $\xi \in \R$, there exists $\xi_0\in\R$ such that
  \[
  \frac{h(\xi)}\xi = \frac{h(\xi) -h(0)}{\xi -0} = h'(\xi_0),
 \]
which implies
\[
\sup_{\xi \in \R} \left|\frac{h(\xi)}\xi\right| \le \sup_{\xi \in \R} |h'(\xi)|. 
\]
 Since $|e^{i m \xi }|= 1$, we have the identity
 \[
 h(\xi) = |\ff(\xi)- e^{-i m \xi }| = |\ff(\xi)e^{i m \xi } -1|,
 \]
 and 
 \[
 | h'(\xi)| \le \left|\ff'(\xi) e^{i m \xi } +im\ff(\xi) e^{i m \xi }  \right| = \left|\ff'(\xi) +im\ff(\xi)   \right|.
 \]
We therefore obtain
\be\label{new-M}
\frac 1m \sup_{\xi \in \R} \frac{|\ff(\xi)- e^{-i m \xi }|}{|\xi|} \le \frac 1m \sup_{\xi \in \R}\left|\ff'(\xi) +im\ff(\xi)   \right| = \sup_{\xi \in \R}\left|\ff(\xi) -\frac{\ff'(\xi)}{\ff'(0)}  \right|.
\ee
Using the argument leading to the upper bound in  \fer{bounds}  one easily concludes that  $H_\infty(F)$ is bounded above by $T(F)$, where $T(F)$ is the functional defined by \fer{ine-T}, and that $H_\infty(F)\le 1$.

\section{A new Fourier-based index of inequality}\label{sec:new}

This Section  will  be devoted to study in more details the main properties of the inequality index $T(F)$, as given by \fer{ine-T}. Depending on convenience, given a random variable $X$ with probability measure $F \in \tilde P_s(\R)$,  we will write indifferently $T(X)$ or $T(F)$.

As outlined in the introduction, for any constant $c>0$, $T(F)$ is  invariant with respect to the scaling $F(x) \to F(cx)$.
The scaling invariance of $T(F)$ can be easily seen by noticing that, if $\ff(\xi)$ is the Fourier transform of $F(x)$, $\ff(\xi/c)$ is the Fourier transform of $F(cx)$, and
\[
\frac{\ff'_\xi(\xi/c)}{\ff'_\xi(0)} = \left. \frac{\ff'_\eta(\eta)}{\ff'_\eta(0)} \right|_{\eta = \xi/c},
\]
that implies
\be\label{sca}
\sup_{\xi \in \R} \left| \ff(\xi/c) - \frac{\ff'_\xi(\xi/c)}{\ff'_\xi(0)} \right| =
\sup_{\xi \in \R} \left| \ff(\eta) - \frac{\ff'_\eta(\eta)}{\ff'_\eta(0)} \right|_{\eta = \xi/c} = T(F).
 \ee
Moreover, as shown in the introduction, if $F \in \tilde P_s(\R)$, $T(F)$ the values of the functional lie  between zero and one, where the value zero (minimal inequality) is assumed in correspondence to a Heaviside probability measure $F_m$, with $m >0$.
In addition, the value $T(F) = 1$, corresponding to maximal inequality is approached if we compute the value of $T(X)$ when, as discussed in Section \ref{sec:H-dist},  the random variable $X$   of mean value $m$ is a two-valued random variable where
\[
P(X= 0) = 1-\e; \quad P\left(X= \frac m\e\right) = \e ; \quad \e \ll 1.
\]
In this case 
\[
\ff(\xi) = 1-\e + \e \exp\left\{-i\frac m\e \xi\right\},
\]
while
\[
\ff'(\xi) = -im \exp\left\{-i\frac m\e \xi\right\}.
\]
Therefore
\begin{equations}\nonumber
&2\,T(X) =  \sup_{\xi \in \R} \left| 1-\e + \e \exp\left\{-i\frac m\e \xi\right\} - \exp\left\{-i\frac m\e \xi\right\}\right| = \\
& (1-\e) \sup_{\xi \in \R} \left| 1 - \exp\left\{-i\frac m\e \xi\right\}\right| = (1-\e)\sup_{\xi \in \R}\sqrt{2(1 - \cos\frac m\e\xi)} =2(1-\e),
\end{equations}
and $T(X) = 1-\e$.

Let $F,G\in \tilde P_s(\R)$ two probability measures with the same mean value, say $m$. Then, for any given $\t \in (0,1)$ it holds $ \t \ff'(0) +(1-\t)\gg'(0) = \ff'(0) = \gg'(0)$, so that
\begin{equations}\nonumber
&T(\t F + (1-\t)G) =\frac 12 \sup_{\xi \in \R} \left| \t \ff(\xi) +(1-\t)\gg(\xi) - \frac{\t \ff'(\xi) +(1-\t)\gg'(\xi)}{\t \ff'(0) +(1-\t)\gg'(0)} \right|  = \\
&\frac 12 \sup_{\xi \in \R} \left| \t \ff(\xi)  - \t \frac{\ff'(\xi) }{ \ff'(0)} + (1-\t) \gg(\xi)  - (1-\t)\frac{ \gg'(\xi) }{ \gg'(0)}\right| \le \t\, T(F) + (1-\t)T(G).
 \end{equations}
 This shows the convexity of the functional $T$ on the set of probability measures with the same mean.
 
However,  the most important property characterizing the inequality index $T$ is linked to its behavior in presence of convolutions. For any given pair of Fourier transforms of probability measures in $\tilde P_s(\R)$, let us set
\[
\hh(\xi) = \ff(\xi)\gg(\xi).
\]
Then, since $|\ff(\xi)| \le \ff(0) =1$ and $|\gg(\xi)| \le \gg(0) =1$
\begin{equations}\nonumber
&\sup_{\xi \in \R} \left| \hh(\xi) - \frac{\hh'(\xi)}{\hh'(0)} \right| = \sup_{\xi \in \R} \left| \ff(\xi)\gg(\xi) - \frac{\ff'(\xi)\gg(\xi) +\ff(\xi)\gg'(\xi)}{\ff'(0)+\gg'(0)} \right| =\\
&\sup_{\xi \in \R} \left| \ff(\xi)\gg(\xi) - \frac{\ff'(0)}{\ff'(0)+\gg'(0)}\frac{\ff'(\xi)\gg(\xi)}{\ff'(0)}- \frac{\gg'(0)}{\ff'(0)+\gg'(0)}\frac{\ff(\xi)\gg'(\xi)}{\gg'(0)} \right| \le\\
&\frac{\ff'(0)}{\ff'(0)+\gg'(0)} \sup_{\xi \in \R} \left| \ff(\xi)\gg(\xi) - \frac{\ff'(\xi)\gg(\xi)}{\ff'(0)}\right| + \\
&\frac{\gg'(0)}{\ff'(0)+\gg'(0)} \sup_{\xi \in \R} \left| \ff(\xi)\gg(\xi) - \frac{\ff(\xi)\gg'(\xi)}{\gg'(0)}\right| \le\\
& \frac{\ff'(0)}{\ff'(0)+\gg'(0)} \sup_{\xi \in \R} \left| \ff(\xi) - \frac{\ff'(\xi)}{\ff'(0)}\right| + \frac{\ff'(0)}{\ff'(0)+\gg'(0)} \sup_{\xi \in \R} \left| \gg(\xi) - \frac{\gg'(\xi)}{\gg'(0)}\right|. 
\end{equations}
Therefore, if $X$ and $Y$ are independent random variables with probability measures in $\tilde P_s(\R)$, and mean values $m_X$ (respectively $m_Y$) the inequality index $T$ satisfies the inequality
\be\label{new-con}
T(X+Y) \le \frac{m_X}{m_X+ m_Y} T(X) + \frac{m_Y}{m_X+ m_Y} T(Y).
\ee
In particular, if $Y$ is a random variable that takes the value $m >0$ with probability $1$ (so that $\gg(\xi) = e^{-im\xi}$ and $T(Y) = 0$), 
\be\label{tide}
T(X+Y) \le \frac{m_X}{m_X+ m} T(X) < T(X).
\ee
Since $X+Y$ corresponds to adding the constant $m$ to $X$, this property asserts that adding a constant wealth to each agent decreases inequality. 

Also, if the random variables $X_1$ and $X_2$ are distributed with the same law of $X$, thanks to the scale property 
\be\label{clt}
T\left(\frac{X_1+X_2}2\right) = T\left(X_1+X_2\right) \le T(X),
\ee
while the mean of $(X_1+X_2)/2$ is equal to the mean of $X$.

\begin{remark} 
Inequality \fer{clt} is fully operational in the case where the two variables $X_1$ and $X_2$ are characterized either by a continuous probability measure or  take on an infinite number of values. 
Only in this case, in fact, do the probability measure  remain of the same type under the operation of convolution.  

Suppose in fact that the variables $X_i$, $i =1,2$, are Bernoulli variables, such that 
\[
P(X_i =0) = P(X_i=1) = \frac 12, \quad i= 1,2,
\]
The probability measure of $X_i$, $i =1,2$, has Fourier transform
\[
\ff(\xi) = \frac 12\left(1 + e^{-i \,\xi}\right),
\]
and the probability measure of the convolution corresponds to the Fourier transform
\[
\ff(\xi)^2 = \frac 14\left( 1+ 2 e^{-i\,\xi} + e^{-2i\,\xi} \right).
\]
Hence, the random variable $Y =X_1+X_2$ takes the three values $0,1,2$ with probabilities
\[
P(Y =0) = P(Y=2) = \frac 14, \quad P(Y =1) = \frac 12.
\]
Clearly, it makes little sense to relate the heterogeneity of a two-valued random variable to a three-valued random variable.
\end{remark} 

Another important consequence of inequality \fer{new-con} is related to the situation in which the random variable $Y$ represents a noise (of  mean value $m>0$) that is present when measuring the inequality index of $X$.  The classical choice is that the additive noise is represented by a Gaussian variable of mean $m$ and variance $\sigma^2$. 

If this is the case, the Fourier transform of the Gaussian density is given by \fer{fou-G},
which is such that
\[
\ff'(\xi) = (-im -\sigma^2\xi) \ff(\xi);  \quad \ff'(0) = -im. 
\]
Hence, since $|\ff(\xi)| \le \ff(0) = 1$,
\[
\left| \ff(\xi) - \frac{\ff'(\xi)}{\ff'(0)}  \right| = \left| \ff(\xi) - \frac{-im -\sigma^2\xi}{-im} \ff(\xi) \right| = \frac\sigma{m} \left| \sigma\xi \ff(\xi)\right|.
\]
Finally, if $Y$ denotes the Gaussian random variable of mean $m >0$ and variance $\sigma^2$ we obtain
\[
T(Y) = \frac\sigma{2m} \sup_{\xi\in \R} \left| \xi e^{-\xi^2/2}\right| = \frac\sigma{2m}\frac 1{\sqrt e}.
\]
As we showed in Section \ref{sec:H-dist} for the index $H$ defined by \fer{H-F}, for a Gaussian variable, the inequality index $T(Y)$ is proportional to the coefficient of variation of $Y$. 
We have in this case
\be\label{new-gau}
T(X+Y) \le \frac{m_X}{m_X+ m} T(X) + \frac{\sigma}{m_X+ m}\frac 1{\sqrt e},
\ee
namely an explicit upper bound for the inequality index in terms of the mean value and the variance of the Gaussian noise.

\section{Examples} \label{sec:examples}

In this section we will recover the values of the inequality index $T$ for some well-known probability measures.  With few exceptions, any time the explicit expression of the Fourier transform of the probability measure is available, the computation of the value of the inequality index $T(\cdot)$ is straightforward.  The list of probability measures that can be treated via Fourier transform is consistent, and includes both discrete and continuous distributions. For an in-depth look at this topic, the interested reader can consult the book \cite{Ober}. 

We do not consider in this paper the possibility to make use of the fast Fourier transform to compute the values of the functional $T$ in the case of a random variable taking only a finite number of values, a situation that we intend to treat in a companion paper. 

\subsection{Two-valued random variables} Let $X$ be a Bernoulli random variable, characterized by the probability measure with Fourier transform
\[
\ff(\xi) = 1-p + p e^{-i\xi}, \quad 0<p<1.
\]
Then, since $\ff'(\xi) = -ip  e^{-i\xi}$, and it is immediate to conclude that
\be\label{Ber}
T(X) = \frac 12 (1-p) \sup_{\xi \in \R} \left|1- e^{-i\xi}\right| = 1-p.
\ee
For given positive constants $a,b$,  let $Y = aX+b$: Then $Y$ is characterized by the Fourier transform
\[
\hh(\xi) = \ff(a\xi) e^{-ib\,\xi}. 
\]
We have
\begin{equations}\nonumber
&\left| \hh(\xi) - \frac{\hh'(\xi)}{\hh'(0)} \right| = \left|\frac{ap \ff(a\xi) - \ff'(a\xi)}{ap+b}  \right| = \frac{ap}{ap+b}(1-p) \left|1- e^{-i\xi}\right|, \\
\end{equations}
so that
\[
T(Y) = \frac{ap(1-p)}{ap+b}.
\]
Choosing $\alpha = b$ and $\beta = a+b$, where $\beta >\alpha$, we then conclude that a two valued random variable $Y$ such that
\[
P(Y = \alpha)  = 1-p, \quad P(Y = \beta ) = p  
\]
has an inequality index
\be\label{two-v}
T(Y) = \frac{(\beta -\alpha) p(1-p)}{\alpha(1-p) + \beta p}.
\ee
The same value is assumed by the Gini and Pietra indices of $Y$.

\subsection{Poisson distribution}
Poisson distribution is characterized by the Fourier transform
\[
\ff(\xi) = \exp\left\{ \lambda\left( e^{-i\xi} -1 \right)\right\} .
\]
In this case
\begin{equations}\nonumber
& \left| \ff(\xi) - \frac{\ff'(\xi)}{\ff'(0)}  \right| = \left|\left( e^{-i\xi} -1 \right) \ff(\xi) \right| = \sqrt{2(1-\cos\xi)}\exp\left\{ -\lambda(1-\cos\xi)\right\}.
 \end{equations}
Let us set $0\le 1-\cos\xi = x^2 \le 2$. Then 
\[
T(F) = \frac {\sqrt 2}2 \sup_{0\le x\le\sqrt 2} x\, e^{-\lambda x^2}.
\]
If $\lambda \le 1/4$, the maximum is taken in $\bar x = \sqrt 2$, and $T(F) = e^{-2\lambda}$. 
If $\lambda > 1/4$, the maximum is taken at the point $\bar x = 1/\sqrt{2\lambda}$, and in this case
\[
T(F) =  \frac 1{2\sqrt\lambda} e^{-1/2}.
\]
Hence, if $F$ is a Poisson probability measure of mean $\lambda$ we have 
 \be\label{coe-P}
T(F) = \left\{
  \begin{array}{cc}
 & e^{-2\lambda} \quad \rm{if}\,\,\,  \lambda \le \frac 14 \\
   & \frac 1{2\sqrt\lambda} e^{-1/2} \quad \rm{if}\,\,\,  \lambda > \frac 14.
  \end{array}
 \right.
 \ee
Note that, as a function of $\lambda$, the functional $T(F)$ is differentiable at the point $\lambda = 1/4$, and it decreases as $\lambda$ increases. Hence, small values of $\lambda$ corresponds to large heterogeneity. 

\begin{remark} It is interesting to remark that the value of the Gini index of a Poisson distribution, say $F$, can not be computed explicitly by resorting to its expression in Fourier transform, as given by formula \fer{Gin-F}. The same conclusion holds if we try to compute the values of $H(F)$, as given by \fer{H-F}, and $H_\infty(F)$, defined  in \fer{sup-F}. 
\end{remark} 

\begin{remark} The previous computations can be extended,  at the cost of more complicated calculations, to evaluate the explicit values of the index $T$ to distributions which are obtained by summing up independent Poisson variables. Maybe the most interesting case corresponds  to the Skellam distribution \cite{Ske,Ske2}, that is the discrete probability distribution of the difference  of two  independent random variables $X_1$ and $X_2$, each Poisson-distributed with  expected values $\lambda_1$
and, respectively $\lambda_2$, with $\lambda_1\not= \lambda_2$. 
\end{remark}

\subsection{Stable laws}
As further example of probability measures defined on the whole real line $\R$, we will compute the value of $T$ in correspondence to a stable law \cite{Zol}. We will restrict here to the case of symmetric alpha-stable distributions of scale parameter $\sigma >0$ and shift parameter $m>0$, characterized by the Fourier transform
\[
\ff_\alpha(\xi) = \exp\left\{ -i\xi\, m - \left| \sigma \xi \right|^\alpha \right\}, \qquad \alpha >1.
\]
Note that the Gaussian distribution of mean $m$ and variance $2\sigma^2$  corresponds to the choice $\alpha = 2$.

For these distributions
\begin{equations}\nonumber
& \left| \ff_\alpha(\xi) - \frac{\ff_\alpha'(\xi)}{\ff_\alpha'(0)}  \right| = \frac \alpha{m}\left| \sigma^\alpha|\xi|^{\alpha -2}\xi  \ff_\alpha(\xi) \right| = \frac \alpha{m}\sigma \left|\sigma\xi\right|^{\alpha -1} \exp\left\{ - \left|\sigma\xi\right|^{\alpha}\right\}.
 \end{equations}
Consequently 
\[
T(F_\alpha) =  \frac{\sigma\alpha}{2\, m} \sup_{x \ge 0} x^{(\alpha -1)/\alpha} e^{-x} 
\]
Evaluating the value of the supremum, we obtain
\be\label{stable}
T(F_\alpha) =  \frac{\sigma\alpha}{2\, m}\left( \frac{\alpha -1}\alpha\right)^{(\alpha -1)/\alpha} \exp\left\{ - \frac{\alpha -1}\alpha\right\}
\ee
For 
$
\alpha =1 $
the distribution reduces to a Cauchy distribution with scale parameter $\sigma$ and shift parameter $m$. In this case
\[
T(F_1) =  \frac{\sigma}{2\, m} \sup_{x \ge 0}  e^{-x}  = \frac{\sigma}{2\, m}.
\]

\subsection{An interesting case: the uniform distribution}

The uniform distribution in the interval  $(-a,a)$, with $a >0$ is characterized by the Fourier transform
\be\label{uni}
\ff(\xi) = \frac{\sin(a\xi)}{a\xi}.
\ee
Hence, if $X$ is a random variable uniformly distributed on $(-a,a)$, for any constant $b>0$, $X+b$ is uniformly distributed on the interval $(-a +b, a+b)$, and the Fourier transform of the probability measure of $X+b$, of mean value $b$ is given by
\[
\gg(\xi) = \ff(\xi)e^{-ib\xi}.
\]
Then
\begin{equations} \nonumber
&\sup_{\xi \in \R} \left| \gg(\xi) - \frac{\gg'(\xi)}{\gg'(0)} \right| = \sup_{\xi \in \R} \left| \ff(\xi)e^{-ib\xi} - \frac{\ff'(\xi)e^{-ib\xi} -ib \ff(\xi)e^{-ib\xi}}{-ib} \right| = \frac 1b\sup_{\xi \in \R}|\ff'(\xi)|.\\
\end{equations}
Next, since $\ff$ is expressed by \fer{uni}
\[
\ff'(\xi) = \frac{a\xi \cos(a\xi) - \sin(a\xi)}{a\xi^2},
\]
which implies
\[
\sup_{\xi \in \R}|\ff'(\xi)| = a \sup_{\xi \in \R}\left|\frac{\xi \cos\xi - \sin\xi}{\xi^2} \right| = a \, \delta_u
\]
where $\delta_u$ is a  positive constant. Hence, if $X$ is uniformly distributed on the interval $(-a,a)$, and $b >0$
\[
T(X+b) = \frac a{2b} \,\delta_u.
\]
In particular, if $b>a$, by setting $\alpha= b-a$ and $\beta = b+a$, we conclude that, if $Y$ is a random variable uniformly distributed on the interval $(\alpha, \beta) \in \R_+$, it holds
\be\label{fin-u}
T(Y) = \frac{\delta_u}2\,\frac{\beta -\alpha}{\beta + \alpha} .
\ee
In this case, at difference with the Gini index, which takes the explicit value
\[
G(Y) =  \frac13\,\frac{\beta -\alpha}{\beta + \alpha} ,
\]
the value of the coefficient $\delta_u$ can be achieved only numerically.  It is however interesting to remark, in the case of a uniform distribution,  the values of the two indices have deep similarities.

A rough estimation of the constant $\delta_u$ follows by studying the function
\[
u(x) = \frac{\sin x -x\cos x}{x^2}, \quad x \ge 0.
\]
It is immediate to show that any extremal point $\bar x $ of the function $u(x)$ solves the equation
\[
(x^2-2)\sin x - 2x\cos x =0,
\] 
that implies 
\[
\sin\bar x -\bar x \cos\bar x = \frac{\bar x^2}2 \sin\bar x.
\]
Consequently, if $\bar x$ is an extremal point of $u(x)$, 
\[
|u(\bar x)| = \frac 12 |\sin \bar x| \le \frac12.
\]
 Hence  $\delta_u \le 1/2$.

To end this Section, we list in Table $4.1$  the values of the inequality index $T$ for some probability measures in $\R_+$ and $\R$ allowing explicit computations. It is remarkable that the Fourier-based index $T$ is well-adapted to compute the heterogeneity index of discrete probability measures, like the negative binomial distribution, or the geometric distribution, which are explicitly expressible in terms of the Fourier transform. We leave the details of the evaluation to the reader.
\begin{table}
    \begin{center}
   { \renewcommand\arraystretch{2.5} 
        \begin{tabular}{lll|ll}
    Measure & Density & Fourier transform & Index $T(\cdot)$\\ 
    \hline
    Exponential &$\lambda e^{-\lambda\, x}$ & $\left( 1 + i\xi/\lambda\right)^{-1}$& $\frac 1{4}$\\
    Gamma & $\frac 1{\Gamma(k)\theta^k} x^{k-1} e^{-x/\theta}$ &$(1+i\theta\xi)^{-k}$ & $\frac 1{2\sqrt{k}} \left(1+ \frac 1{k}\right)^{-(k+1)/2} $\,\,\, $k > 0$\\
    Chi-squared &$\frac 1{2^{k/2}\Gamma(k/2)}x^{k/2-1} e^{-x/2}$ & $(1+2i\xi)^{-k/2}$ & $\frac 1{\sqrt{2k}} \left(1+ 2/k\right)^{-(k+2)/4} $\,\,\, $k \ge 1$\\
    Laplace &$\frac 1{2\sigma}e^{-|x-m|/\sigma}$ & $e^{-im\,\xi}\left( 1 + \sigma^2\xi^2\right)^{-1}$& $\frac\sigma{m}\frac{16}{25}$\\
    \hline
        \end{tabular}
        }
        \label{table1}
\end{center}
\vskip 2mm
    \caption[]{Values of the index $T$ for some probability measures }
\end{table}
 \section{An application to kinetic theory of wealth distribution}\label{sec:kinetic}
 
 Kinetic modelling of agent-based markets are based on few universal \index{wealth distribution models}
assumptions \cite{PT13}. First, agents are indistinguishable, so that an agent's
{state} at any instant of time $t\geq0$ is completely characterized
by his current wealth $w\geq0$. Second,  the time variation of the wealth
distribution is entirely due to binary trades between agents. A trade
represents a binary interaction in which part of the money of each agent
is modified according to well-defined rules. When two agents undertake in
a trade, their {pre-trade wealths\/} $v$, $w$ change into the {
post-trade wealths\/} $v^*$, $w^*$ according to a linear exchange rule:
\begin{equation}
  \label{eq.trules}
  v^* = p_1 v + q_1 w, \quad w^* = q_2 v + p_2 w.
\end{equation}
The {interaction coefficients\/} $p_i$ and $q_i$, $i=1,2$, are, in
general, non negative random parameters.
 
 The first explicit description of a binary wealth-exchange model
dates back to the seminal work of Angle~\cite{Ang}, (cf. also \cite{Ang1}), even if the
intimate relation to statistical mechanics was only described about
a decade later~\cite{DY00,IKR}. In each binary interaction, winner
and loser are randomly chosen, and the loser pays a random fraction
of his wealth to the winner. From here, Chakraborti and
Chakrabarti~\cite{CC00} developed the class of {\em strictly
conservative\/} exchange models, which preserve the total wealth in
each individual trade,
\begin{equation}
  \label{eq.strict}
  v^* + w^* = v + w .
\end{equation}
In its most basic version, the microscopic interaction is determined by
one single parameter $\lambda\in(0,1)$, which is the global {\em saving \index{wealth distribution models!saving propensity}
propensity}. In the interactions, each agent retains the corresponding fraction
of its pre-trade wealth, while the rest $(1-\lambda)(v+w)$ is equally
shared equally between the two trading partners,
\begin{equation}
  \label{eq.xchg}
  v^* = \lambda v + \frac12(1-\lambda)(v+w), \quad   w^* = \lambda w + \frac12(1-\lambda)(v+w) .
\end{equation}

The wealth distribution $f(v,t)$ of the system of agents coincides with agent's density and satisfies the associated spatially homogeneous Boltzmann
equation,
\begin{equation}
  \label{eq.boltzmann}
 \partial_t f + f = Q_+(f,f),
\end{equation}
on the real half-line, $v\geq0$. The collisional gain operator $Q_+$ acts
on test functions $\varphi(v)$ as
\begin{align}
\nonumber Q_+(f,f)[\varphi] =&
  \int_{\setR_+} \varphi(v)Q_+\big(f,f\big)(v)\,dv\\[-.25cm]
  \label{eq.qweak}
  \\[-.25cm]
  \nonumber
  =& \frac12 \int_{\setR_+^2} \langle \varphi(v^*) + \varphi(w^*) \rangle
  f(v)f(w)\,dv\,dw.
\end{align}
In reason of \fer{eq.xchg}, the average wealth of the society is
conserved with time, so that
\begin{equation}
  \label{eq.conserve}
  m(t) = \int_{\setR_+} w f(w,t)\,dw = m,
\end{equation}
where $m>0$ is finite. 
A useful way of writing equation \fer{eq.boltzmann} is to resort to the Fourier
transform \cite{PT13}. Assuming the initial distribution of wealth in $P_s^+$, with $s
>1$,  the transformed
kernel reads
\begin{equation}
  \label{trascoll}
  \widehat{Q}\bigl( \ff,\ff \bigr)(\xi) = \ff\left( \frac{1-\lambda}2\xi\right) \ff\left( \frac{1+\lambda}2\xi\right),
    \end{equation}
where, since the initial density has a bounded mean,
\[
\widehat{f_0}^\prime(0) = - im .
\]
Hence, the Boltzmann equation  \fer{eq.boltzmann} can be rewritten in terms of the Fourier transform of $f(v,t)$ as 
 \begin{equation}
  \label{fkac}
  \frac{\partial \widehat f(\xi,t)}{\partial t}  + \widehat f(\xi,t) = \ff\left( \frac{1-\lambda}2\xi, t\right) \ff\left( \frac{1+\lambda}2\xi, t\right).
\end{equation}
It is immediate to show that  the functions $e^{-i\mu\xi}$,  with $\mu >0$, namely the Fourier transforms of a Dirac delta concentrated at the  wealth $\mu$, are stationary solutions of equation \fer{fkac}. 

More can be said if we assume that $s \ge 2$. Then, the  moment of order two of the initial distribution is finite, and, applying \fer{trascoll} with $\varphi(v) = (v-m)^2$, and recalling that the mean value is preserved during the evolution, shows that the variance
of $f(v,t)$ satisfies
\begin{equation}
  \frac{d}{dt} \int_{\setR_+}  (v-m)^2 f(v,t)\,dv = - \frac12(1-\lambda^2) \int_{\setR_+}  (v-m)^2 f(v,t)\,dv.
\end{equation}
As a result, all agents tend for large times to become equally rich. Indeed,
the steady state $f_\infty(v)$ is a Dirac delta concentrated
at the mean wealth, and is approached at the exponential rate
$(1-\lambda^2)/2$.

To remain into the framework of inequality indices, the previous result implies that, if the initial distribution belongs to $P_s^+$, with $s \ge 2$, the coefficient of variation is monotonically decreasing towards zero at the explicit rate $(1-\lambda^2)/4$. 

This result is lost as soon as the value of $s$ is less than $2$. It is however interesting to remark that the inequality index $T(F(t))$, where $F(v,t)$ is the probability measure associated to the solution $\ff(\xi,t)$ of equation \fer{fkac}, is monotonically decreasing in time as soon as $s \ge 1$. Indeed, if we set
\[
h(\xi,t) = \ff(\xi, t) - \frac {\ff'(\xi,t)}{\ff'(0)}, 
\] 
it is immediate to show  that $h(\xi,t)$ satisfies the equation
\[
 \frac{\partial h(\xi,t)}{\partial t}  + h(\xi,t) = \ff\left( \frac{1-\lambda}2\xi, t\right) \ff\left( \frac{1+\lambda}2\xi, t\right) - \frac{ \left[\ff\left( \frac{1-\lambda}2\xi, t\right) \ff\left( \frac{1+\lambda}2\xi, t\right)\right]'}{\ff'(0,t)} ,
\]
which implies
\be\label{med}
\left| \frac{\partial h(\xi,t)}{\partial t}  + h(\xi,t)\right| \le \sup_\xi\left|\ff\left( \frac{1-\lambda}2\xi, t\right) \ff\left( \frac{1+\lambda}2\xi, t\right) - \frac{ \left[\ff\left( \frac{1-\lambda}2\xi, t\right) \ff\left( \frac{1+\lambda}2\xi, t\right)\right]'}{\ff'(0,t)}\right|.
\ee
If now $X(t)$ and $Y(t)$ are random variables with probability measures of Fourier transforms $\ff\left( \xi\,(1-\lambda)/2, t\right)$ (respectively $\ff\left( \xi\,(1+\lambda)/2, t\right)$), which have mean values $m\,(1-\lambda)/2$ (respectively $m\,(1+\lambda)/2$), formula \fer{new-con} for convolutions gives
\be\label{vab}
T(X(t) + Y(t)) \le \frac{1-\lambda}2 T(X(t)) + \frac{1-\lambda}2 T(Y(t)).
\ee
On the other hand, by scaling invariance $T(X(t)) = T(Y(t)) = T(Z(t)$, where the probability measure of $Z(t)$ has Fourier transform $\ff(\xi,t)$. Hence, equation \fer{med} implies
\be\label{fin}
\left| \frac{\partial h(\xi,t)}{\partial t}  + h(\xi,t)\right| \le \sup_\xi\left| h(\xi,t)\right|,
\ee
that, for any given $t_0 <t$ by Gronwall inequality implies \cite{PT13}
\[
T(Z(t)) \le T(Z(t_0),
\]
and, consequently the monotonicity in time of the inequality index $T(F(t)$ of the probability measure solution of the kinetic equation \fer{eq.boltzmann}. It is remarkable that this result, which does not require the condition $s >1$, is a direct consequence of the convolution property of the inequality index $T$. Hence, the monotonicity result does not hold if we resort to Gini and Pietra indices.

\section{Conclusions}
Inequality indices are quantitative scores that take values in the unit interval, with the zero score characterizing perfect equality. Measuring the statistical heterogeneity of  measures arises in most fields of science and engineering, which makes it important to know the strengths and possible weaknesses of heterogeneity measures in applications  \cite{Ban, BL, Cou,Cow,Eli, HN, HR}.  In this paper, we draw attention to a new inequality index, based on the Fourier transform, which exhibits a number of interesting properties that make it very promising in applications. In comparison with the well-known and widely used Gini index, which can still be expressed by resorting to Fourier transform, the new index $T$ allows to compute explicitly the heterogeneity of various probability measures, like the Poisson distribution, which can not be measured explicitly resorting to Gini index. Also, this new Fourier-based index has an interesting  property of  sub-additivity for convolutions, which in principle makes it interesting for applications to models of kinetic theory which contain mass and mean preserving bilinear operators \cite{PT13}.


\section*{Acknowledgement} This work has been written within the
activities of GNFM group  of INdAM (National Institute of
High Mathematics), and partially supported by  IMATI (Institute of Applied Mathematics and Information Technologies Enrico Magenes).


\end{document}